\documentclass[12pt,letterpaper,reqno]{amsart}
\usepackage[parfill]{parskip}
\usepackage{amssymb}
\usepackage{amsmath}
\usepackage{amsthm}
\usepackage{amsfonts}
\usepackage[utf8]{inputenc}

\usepackage[T1]{fontenc}
\usepackage{palatino}
\usepackage[text={6.5in, 9in}, centering]{geometry}

\usepackage[usenames]{color}
\newtheorem{theorem}{Theorem}
\newtheorem{definition}[theorem]{Definition}

\newtheorem{lemma}[theorem]{Lemma}

\newtheorem{remark}[theorem]{Remark}

\newtheorem{proposition}[theorem]{Proposition}

\newcommand{\R}{\mathbb{R}}

\newcommand{\supp}{\mbox{supp}~}

\title[3D MHD energy cascades]{Energy cascades in physical scales of 3D incompressible magnetohydrodynamic turbulence}
\author{Z. Bradshaw and Z. Gruji\'c}
\address{Department of Mathematics\\
University of Virginia\\ Charlottesville, VA 22904}
\date{\today}
\begin{document}
\begin{abstract}The existence of a total energy cascade and the scale-locality of the total energy flux are rigorously established working directly from the 3D MHD equations and under assumptions consistent with physical properties of turbulent plasmas.  Secondary results are included identifying scenarios where inertial effects on specific energies effect cascade-like behavior as well as a scenario in which the inter-field energy transfer is predominantly from the fluid to the magnetic field.
\end{abstract}
\maketitle
\section{Introduction}

The 3D magnetohydrodynamic equations (3D MHD) model the evolution of a coupled system comprised of a magnetic field and an electrically conducting fluid's velocity field.  
Throughout turbulent MHD regimes such as stellar winds and the interstellar medium, observational and numerical evidence indicate that energy is transported from larger to smaller scale structures in a regular fashion (cf. \cite{DB,B05,BhaPo10,DB-ES-01}).  This process, known as the \emph{energy cascade}, is consistent with the picture wherein energy is preferentially distributed on intermittently located and progressively thinning coherent current and vortex structures \cite{DB, BrGr2, Greco09, Greco08, Bruno2001,PoBhGa}.
The present contribution presents a mathematical framework (initially developed to study hydrodynamic turbulence, cf.\cite{DaGr1}) by which we quantify and rigorously affirm the existence of an energy cascade across an inertial range as an \emph{intrinsic feature of the 3D MHD system}.  Before further discussing our work, we briefly discuss the topic of energy cascades in general.

A (direct) cascade, whether it be in a fluid or plasma, is the net inertial transport of an ideally conserved quantity from larger to smaller scales (cf. \cite{Chorin, Frisch,DB}).  Roughly put, a source injects energy and, as the medium transitions to turbulence, large scale coherent structures emerge -- vortex filaments in the hydrodynamic case; current sheets in 3D MHD turbulence -- on which this energy is concentrated.  As turbulence evolves, the energy is transported by inertial effects from macro-scale eddies to progressively smaller scales in a uni-directional fashion.  The process ceases at a scale at which  inertial effects are outweighed by dissipative forces (either viscous or resistive) and, instead of being transported to even smaller scales, energy is lost as heat.
The range of scales over which this cascade persists is referred to as the \emph{inertial range}. In the 3D hydrodynamic energy cascade, the inter-scale transport obeys two fundamental properties: constancy and scale locality of the flux.  The first of these means that, at every scale within the inertial range, the energy flux is constant.  Locality of the flux means that the energy exchange is predominantly between structures of comparable scales.  

The existence of an energy cascade in 3D MHD turbulence is widely accepted in the physics community (see \cite{DB} for an overview and \cite{Kr,Ir} for the classical phenomenologies), but there is considerable disagreement regarding the details.  A contentious issue lies in understanding the anisotropic influence of a strong magnetic field on scaling properties of the energy spectrum (this discussion began in earnest in \cite{GS94,GS95}; in contrast, the classical phenomenologies of Iroshnikov and Kraichnan assumed an isotropic spectral transfer \cite{Kr,Ir}).  In \cite{GS95}, a \emph{critical balance} assumption -- i.e. that there is a single timescale for parallel and perpendicular motion (to the magnetic mean field) in a turbulent eddy -- was introduced allowing for the derivations of a distinct perpendicular energy spectrum and a scaling relationship between the lengths of perpendicular and parallel fluctuations.  Numerical results indicated the picture is more complex than that described in \cite{GS95} and various competing phenomenologies have been developed (cf. \cite{B05,MuGr05,GaPoMa05,BeLa08} for several examples). Lively debate remains as to which is the most effective (cf. \cite{PeMaBoCa12,Be12}).  

Although there is potential, in light of the discussion highlighted above, for a contribution to this debate based on a mathematical analysis of the governing system, there have been few results in this direction.  Indeed, prior to now, neither the existence of an energy cascade nor the scale locality of the flux had been rigorously affirmed (results concerning locality of the flux can be found in \cite{AlEy10}; a rigorous study of a related phenomenon, the concentration of enstrophy, has been carried out in \cite{BrGr2}). The purpose of this paper is to provide such results; in Section \ref{sec:cascades}, we establish the cascade of total energy by studying the orientation of the total energy flux in a suitably statistical manner across an inertial range and, in Section \ref{sec:locality}, affirm scale locality of the energy flux. 
We remark that our conclusions follow directly from the 3D MHD equations under assumptions which are physically reasonable for turbulent regimes wherein the magnetic Prandtl number is not significantly smaller than one (i.e. $\eta \lesssim \nu$).  In particular, it applies to astronomical settings such as the Solar wind and the interstellar medium.  Interestingly, and contrasting the current phenomenological theories cited earlier, no appeal is made to the existence of a strong magnetic guide field. This indicates that the energy cascade is an \emph{intrinsic property of the 3D MHD system} and not solely an artifact of the plasma's environment. 

Due to the coupling between the magnetic field and plasma, each of which is imbued with its own energy, there are a number of transfer mechanisms by which energy can `flow' between scales.  In particular, kinetic energy can remain tied to the velocity field or magnetic energy to the magnetic field, but each energy can also be transferred between the two fields.  A secondary purpose of this work is to identify conditions under which the distinct energies are transported, in a statistical sense, from larger to smaller scales (see Sections \ref{sec:fluid} and \ref{sec:indiecascades}).   Additionally, by considering the stretching effect of the velocity field on the magnetic field lines, we identify a scenario and range of scales in which the dominant inter-field energy exchange is directed from the velocity field to the magnetic field (see Section \ref{sec:scenario}).

Our conclusions are obtained using a dynamic, multi-scale averaging process developed to study features of hydrodynamic turbulence (we recall the specifics of this methodology in Section 2; the initial development of this approach appears in \cite{DaGr1} and has subsequently been used in \cite{BrGr2,DaGr2, DaGr3, DaGr4}). The process acts as a detector of significant \emph{sign-fluctuations} associated with a physical density at a given scale and is used to show that the orientation of a particular flux -- i.e. the energy flux -- is, in a statistically significant sense, from larger to smaller scales. The analysis is carried out entirely in physical space.  This contrasts other approaches which have identified `scale' with the Fourier wavenumber (cf. \cite{FMRT2001}). Interestingly, locality is derived \emph{dynamically} as a direct consequence of the existence of the turbulent cascade in view, featuring comparable upper and lower bounds throughout the inertial range; in contrast, the previous locality results were essentially localized \emph{kinematic} upper bounds on the flux, the corresponding lower bounds being consistent with turbulent properties of the flow \cite{E05, CCFS08}.

\section{$(K_1,K_2)$-covers and ensemble averages}\label{sec:covers}

The main purpose of this section is to describe how {\em ensemble averaging} with respect to {\em $(K_1,K_2)$- covers} of an integral domain $B(0,R_0)$ can be used to establish \emph{essential positivity} of an \emph{a priori} sign-varying density over a range of physical scales associated with the integral domain (cf. \cite{DaGr1}).  The application to turbulence lies in showing certain flux densities are directed into structures of a particular scale -- i.e. the cascade is uni-directional from larger to smaller scales -- as well as the near-constancy of the averaged densities -- i.e. the space-time averages over cover elements are all mutually comparable -- across a range of scales.

The ensemble averages will be taken over collections of spatiotemporal averages of physical densities localized to cover elements of a particular type of covering -- a so called {\em $(K_1,K_2)$-cover} -- where the cover is over the region of turbulent activity.  For simplicity, this region will be taken as a ball of radius $R_0$ centered at the origin and, to reflect the turbulence literature, is henceforth referred to as the {\em integral domain} (also known as the \emph{macro-scale domain}).  The time interval on which we localize is  motivated by the physical theories of turbulence and assumed to satisfy, \[T\geq \frac {R_0^2} \nu.\]
The $(K_1,K_2)$-covers are now defined.

\begin{definition}Let $K_1,K_2\in \mathbb N$ and $0\leq R\leq R_0$.  The cover of the integral domain $B(0,R_0)$ by the $n$ (open) balls, $\{B(x_i,R)\}_{i=1}^n$ is a  {\em $(K_1,K_2)$-cover at scale }$R$ if,
\begin{align*}\bigg( \frac {R_0} R \bigg)^3\leq n \leq K_1 \bigg(\frac {R_0} R\bigg)^3\end{align*}
and, for any $x\in B(0,R_0)$, $x$ is contained in at most $K_2$ cover elements.
\end{definition}

In the hereafter, all covers are understood to be $(K_1,K_2)$-covers at scale $R$.
The positive integers $K_1$ and $K_2$ represent the maximum allowed \emph{global} and \emph{local multiplicities}, respectively.

In order to localize a physical density to a cover element we incorporate certain {\em refined} cut-off functions.  For a cover element centered at $x_i$, let $\phi_i(x,t)=\eta(t)\psi(x)$ where $\eta\in C^\infty(0,T)$ and $\psi\in C_0^\infty (\R^3)$ satisfy,
\begin{align}\label{timecutoff}0\leq \eta\leq 1,\qquad \eta=0~\mbox{on }(0,T/3),\qquad\eta=1~\mbox{ on }(2T/3,T),\qquad\frac {|\partial_t\eta|} {\eta^\delta }\leq \frac {C_0} T,
\end{align}
and,
\begin{align}\label{spacecutoff} 0\leq \psi\leq 1,\qquad\psi=1\mbox{ on }B(x_i,R),\qquad\frac {|\partial_i\psi|} {\psi^{\rho}} \leq \frac {C_0} {R},\qquad\frac {|\partial_i\partial_j \psi|} {\psi^{2\rho-1}}\leq \frac {C_0} {R^2},
\end{align}where $3/4< \delta,\rho <1$.

By $\phi_0$ we denote a refined cut-off function centered at $x=0$ localizing to the ball $B(0,R_0)$ (note that $\phi_0$ is the cut-off function for the integral domain).

Comparisons will be necessary between averaged quantities localized to cover elements at certain scales $R<R_0$ and averaged quantities taken at the scale of the integral domain, $R_0$.  To accommodate this we impose several additional conditions for points $x_i$ lying near the boundary of $B(0,R_0)$.  If $B(x_i,R)\subset B(0,R_0)$ we assume $\psi\leq \psi_0$.  Alternatively, when $B(x_i,R)\not\subset B(0,R_0)$, additional assumptions are in order. To specify these, let $l(x,y)$ denote the collection of points on the line through $x$ and $y$ and define the sets,  
\begin{align*}
S_0&=B(x_i,R)\cap B(0,R_0),
\\S_1&=\big\{x: R_0\leq |x|< 2R_0 ~\mbox{and}~\emptyset\neq \big( l(x,0)\cap  \partial B(x_i,R)\cap B(0,R_0)^c \big)\big\},
\\S_2&=\bigg( B(x_i,2R)\cup \big\{x: R_0\leq |x|< 2R_0 ~\mbox{and}~\emptyset\neq \big( l(x,0)\cap  \partial B(x_i,2R)\cap B(0,R_0)^c \big)\big\}\bigg)\setminus (S_0\cup S_1).
\end{align*}Then, our assumptions are that $\psi$ satisfies (\ref{spacecutoff}), $\psi=1$ on $S_0$, $\psi=\psi_0$ on $S_1$, and $\supp \psi = S_2$.  The above conditions ensure that $\psi\leq \psi_0$ and that $\psi$ can be constructed to have an inwardly oriented gradient field.

The above apparatus is employed to study properties of a physical density at a {\em physical scale} $R$ associated with the integral domain $B(0,R_0)$ in a manner which we now illustrate.  Let $\theta$ be a physical density (e.g. a flux density) and define its localized spatio-temporal average on a cover element at scale $R$ around $x_i$ as \begin{align*}\tilde \Theta_{x_i,R} = \frac 1 T \int_0^T \frac 1 {R^3} \int_{B(x_i,2R)} \theta(x,t)\phi^\delta_{i}(x,t)~dx~dt,\end{align*}
where $0< \delta\leq 1$, and let $\langle \Theta\rangle_R$ denote the ensemble average over localized averages associated with cover elements,
\[\langle \Theta \rangle_R=\frac 1 n \sum_{i=1}^n \tilde \Theta_{x_i,R}.\]
Examining the values obtained by ensemble averaging the averages associated to a variety of covers at a fixed  scale allows us to draw conclusions about the flux density $\theta$ at comparable and greater scales.  For instance, stability (i.e. near constancy) across the set  $\{\langle \Theta\rangle _R\}$ indicates that the sign of $\theta$ is essentially uniform at scales comparable to or greater than $R$.  On the other hand, if the sign were not essentially uniform at scale $R$, particular covers could be arranged to enhance negative and positive regions and thus give a wide range of sign varying values in $\{\langle \Theta\rangle _R\}$.  Our methodology, then, establishes the essential positivity of an {\em a priori} sign varying density $\theta$ at a scale $R$, by showing the positivity and near constancy of all elements of $\{\langle \Theta\rangle _R\}$.

An indispensable observation is that, if $\theta$ is an {\em a priori} non-negative density, then the ensemble averages taken at scales below the integral scale are all comparable to the integral scale average.  We make this notion precise in the following lemma.
\begin{lemma}\label{lemma:ensembleaverages}Let $f(x,t)\in L_{loc}^1((0,T)\times \R^3)$ be non-negative.  Let $\{x_i\}_{i=1}^n$ be centers of elements of a $(K_1,K_2)$-cover of $B(x_0,R_0)$ at scale $R<R_0$.  Setting \[F_0=\frac 1 T \int_0^T \frac 1 {R_0^3} \int f(x,t)\phi_0(x,t) ~dx~dt,\]
and \[F_{x_i,R}=\frac 1 T \int_0^T \frac 1 {R^3} \int f(x,t)\phi_{x_i,R}(x,t) ~dx~dt,\]we have
\begin{align}\label{ineq:interpBetweenScales} \frac 1 {K_1} F_0 &\leq \langle F  \rangle_R \leq K_2 F_0. \end{align}
\end{lemma}
\begin{proof}Recalling that $\phi_{x_i,R}\leq \phi_0$ and the definition of $(K_1,K_2)$-covers we have that, \begin{align*}\langle F  \rangle_R &=\frac 1 T \int_0^T \frac 1 {nR^3} \int f(x,t) \sum_{i=1}^n \phi_{x_i,R}(x,t) ~dx~dt
\\&\leq \frac 1 T \int_0^T \frac 1 {R^3} \frac {R^3} {R_0^3} \int f(x,t) K_2 \phi_0(x,t) ~dx~dt = K_2F_0,\end{align*}
and,
\begin{align*}\langle F  \rangle_R &=\frac 1 T \int_0^T \frac 1 {nR^3} \int f(x,t) \sum_{i=1}^n \phi_{x_i,R}(x,t) ~dx~dt
\\&\geq \frac 1 T \int_0^T \frac 1 {R^3} \frac 1 {K_1} \frac {R^3} {R_0^3} \int f(x,t)  \phi_0(x,t) ~dx~dt =\frac 1 {K_1} F_0.
\end{align*}
\end{proof}

For additional discussion of $(K_1,K_2)$-covers and ensemble averages, including some computational illustrations of the process, see \cite{DaGr3}.

%*********************************************************************
\section{3D incompressible MHD equations}\label{sec:context}
%*********************************************************************

Our mathematical setting is that of {\em weak solutions} to the 3D magnetohydrodynamic equations over $\R^3$ (cf. \cite{SeTe} for the essential theory).  Define $\mathcal V=\{f\in  L^2(\R^3):\nabla\cdot f=0\}$ (where the divergence free condition is in the sense of distributions) and let $V$ be the closure of $\mathcal V$ under the norm of the Sobolev space, $(H^1(\R^3))^3$, and, $H$, the closure of $\mathcal V$ under the $L^2$ norm. By a solution to 3D MHD we mean a weak (distributional) solution to the following coupled system (3D MHD):
\begin{align*}u_t-\nu \triangle u +(u\cdot \nabla)u- (b\cdot \nabla)b +\nabla (p+|b|^2/2) &=0,
\\ b_t-\eta \triangle b + (u\cdot \nabla) b - (b\cdot \nabla )u &=0 ,
\\ \nabla\cdot u = \nabla\cdot b&= 0,
\\ u(x,0)=u_0(x)&\in V,
\\ b(x,0)=b_0(x) & \in V,
\end{align*}
where $\eta$ and $\nu$ are the magnetic resistivity and kinematic viscosity respectively and $p(x,t)$ is the fluid pressure.

Our present work utilizes {\em suitable weak solutions} for MHD. These are weak solutions which additionally satisfy the generalized energy inequality (among other things -- see \cite{HeXin1} for a precise definition),
\begin{align}\label{ineq:generalizedEnergy} &\int_0^T\int ( \nu |\nabla u(x,t)|^2+\eta |\nabla b(x,t)|^2)\phi(x,t)~dx~dt
\\\notag &\qquad \leq \quad\frac 1 2\int_0^T\int(|u(x,t)|^2+|b(x,t)|^2)\phi_t(x,t)~dx~dt
\\\notag&\qquad\qquad+ \frac  1 2\int_0^T\int (\nu |u(x,t)|^2+\eta |b(x,t)|^2)\Delta \phi(x,t)~dx~dt
\\\notag &\qquad\qquad+ \frac  1 2\int_0^T\int (|u(x,t)|^2+|b(x,t)|^2+2 p(x,t))(u(x,t)\cdot\nabla \phi(x,t))~dx~dt
\\\notag &\qquad\qquad -\int_0^T \int (u(x,t)\cdot b(x,t))(b(x,t)\cdot\nabla \phi(x,t))~dx~dt,
\end{align} for a.e. $T\in (0,\infty)$ and any non-negative $\phi\in C_0^\infty (\R^3\times [0,\infty))$.

Existence of suitable weak solutions for  MHD is proven in \cite{HeXin1} using an adaptation of the traditional method for NSE found in \cite{CKN-82}.  Our application of these will require the generalized energy inequality mentioned above as well as the following regularity properties (these can also be found in \cite{HeXin1}).

\begin{proposition}For $u_0$, $b_0\in H$ and $u_0\in W^{4/5,5/3}$, suppose $(u,b,p)$ constitutes a suitable weak solution to 3D MHD.  Then, $(u,b,p)$ satisfies,
\begin{align*}u,b\in L^\infty(0,\infty;H),&~ u,b\in L^{3}(0,\infty;L^3(\R^3)),
\\ \nabla u,\nabla b\in L^2(0,\infty;L^2(\R^3)),&~ p\in L^{3/2} (0,\infty;L^{3/2}(\R^3)).
\end{align*}
\end{proposition}

The fact that suitable weak solutions only satisfy a generalized energy  inequality (as opposed to equality) introduces the possibility that energy is dissipated not only by viscosity or resistivity but also by singularities.  In the case that the weak solution in question is {\em regular}, equality is attained in the generalized energy inequality, \eqref{ineq:generalizedEnergy}, and the potential for loss of flux due to singularities is eliminated.  To streamline discussion we establish cascades for regular solutions and include an illustrative result for the non-regular case only in the context of the cascade for the modified (due to energy loss from possible singularities) total energy flux (cf. Section \ref{sec:cascades}).  The study of the possible energy loss due to singularities is itself an interesting subject but the case of MHD is not sufficiently distinct from that of NSE (which can be found in \cite{DaGr1}) to justify an independent exposition.

%*******************************************************************************
%*******************************************************************************
%
% TOTAL ENERGY
%
%*******************************************************************************
%*******************************************************************************
\section{Total energy cascade}\label{sec:cascades}

The total energy flux through the boundary of the ball $B$ over the interval $(0,T)$ is given (cf. \cite{DB}) by,
\[\frac 1 2 \int_0^T \int_{\partial B}(|u|^2+|b|^2+2 p)~\hat n \cdot u~dx~dt -\int_0^T \int_{\partial B}(u\cdot b) ( \hat n \cdot b)~dx~dt,\]where $\hat n$ is the unit normal vector directed inward.  Our analytic results are enabled by substituting an inwardly directed vector field, $\nabla \phi$, for $\hat n$ where $\phi$ is a refined cut-off function for the ball $B(x_0,R)$.  Our localized, space-time averaged total energy flux into the ball centered at the spatial point $x_0$ of radius $R$, over the interval $(0,T)$, is then defined as,
\begin{align*}\label{def:totalEnergyFlux} F^{E}_{x_0,R} &:=\frac 1 2 \int_0^T\int (|u|^2+|b|^2+2 p)(u\cdot\nabla \phi)~dx~dt
-\int_0^T \int (u\cdot b)(b\cdot\nabla \phi)~dx~dt.  \end{align*}

The following remark on the genesis of the last term, the advection of the cross-helicity via the magnetic field, is informative. The transfer of magnetic to kinetic energy is driven by the Lorentz force while the stretching of the magnetic field lines is responsible for the transfer of kinetic energy to magnetic energy.  Since these are complementary, the sum cancels in the non-localized case.  Locally, however, we obtain a flux-type term,
\begin{align}\int_0^T \int  ((b\cdot \nabla )b \cdot \phi u  + (b\cdot\nabla)u \cdot \phi b )~dx~dt= -\int_0^T\int (u\cdot b )(b\cdot \nabla \phi)~dx~dt.
\end{align} The appearance of this term is interesting because it is only non-zero if there is some degree of non-locality in the energy transfer \emph{between the two fields}.

In the following we work in a fixed integral domain, $B(x_0,R_0)$, with associated cut-off $\phi_0=\phi_{x_0,R_0}$.  Certain integral domain quantities will be used to determine lower bounds on the inertial ranges over which our cascades are shown to persist.  These are the integral scale space-time averaged kinetic and magnetic energies which are defined in terms of a technical parameter, $\delta$, as,
\begin{align*}e_0^u=e_0^u(\delta)&=  \frac 1 T \int_0^T \frac 1 {R_0^3} \int \frac 1 2|u |^2\phi_0^\delta  ~dx~dt\quad \mbox{and}\quad e_0^b=e_0^b(\delta)=  \frac 1 T \int_0^T \frac 1 {R_0^3} \int \frac 1 2|b |^2\phi_0^\delta  ~dx~dt,\end{align*}
and the integral scale space-time averaged enstrophies which are given by,
\begin{align*}E_0^u= \nu\frac 1 T \int_0^T \frac 1 {R_0^3}\int |\nabla u|^2\phi_0~dx~dt\quad \mbox{and}\quad E_0^b= \eta\frac 1 T \int_0^T \frac 1 {R_0^3}\int |\nabla b|^2\phi_0~dx~dt.
\end{align*}
The combined kinetic and magnetic energies or enstrophies will be identified by omitting the superscript (i.e. $e_0:=e_0^u+e_0^b$). Note that, because $e_0$ is decreasing with $\delta$, we will take liberties suppressing the dependence of $e_0$ on $\delta$ with the understanding that the indicated quantity is that associated with the smallest appropriate value for $\delta$.

We will establish that the cascade persists over a range of scales bound above by the integral scale and below by a modified Taylor micro-scale.  The Taylor micro-scale is (cf. \cite{DB}),
\[\tau = \bigg(\frac {\nu e_0} {E_0}\bigg)^{1/2}.\]
The modification will depend in part on the {\em magnetic Prandtl number}, denoted by $Pr$, a non-dimensional number given by the ratio of kinematic viscosity to magnetic resistivity, i.e.,
\[Pr=\frac \nu \eta.\]Essentially, its role in the modification is to incorporate information about the magnetic resistivity which is absent from our prescribed time scale (recall $T>R_0^2/\nu$) and the Taylor micro-scale.

\begin{theorem}\label{thrm:totalCascade}Assume $u$ and $b$ are suitable weak solutions to 3D MHD with sufficient regularity that equality holds in \eqref{ineq:generalizedEnergy}. Let $\{x_i\}_{i=1}^n\subset B(x_0,R_0)$ be the centers of a $(K_1,K_2)$-cover at scale $R$.  For a scale and cover independent positive parameter, \[\beta:=\bigg( \frac {1} {2CK_1K_2(1+Pr^{-1})}\bigg)^{1/2},\] where $C$ is a constant determined by structural properties of 3D MHD and our cut-off functions, if
$\tau/ \beta< R_0$
then,
 \[\frac {1} {2K_1} E_0 \leq \big\langle   F_{}^E \big\rangle_R \leq  2K_2 E_0,\]
 provided $R$ is contained in the interval
$[\tau /\beta, R_0]$.
\end{theorem}

\begin{proof}Assuming the premises above and in virtue of \eqref{ineq:generalizedEnergy}, we have for any cover element that, \begin{align*}F_{x_i,R}^E&\geq \int_0^T\int (\nu |\nabla u |^2+\eta |\nabla b |^2)\phi_{x_i,R} ~dx~dt
\\&\qquad - \bigg| \frac  1 2\int_0^T\int(|u|^2+|b|^2)\partial_t\phi_{x_i,R}~dx~dt  + \frac  1 2\int_0^T\int(\nu|u|^2+\eta |b|^2)\Delta \phi_{x_i,R} ~dx~dt \bigg|.
\end{align*}
Recalling the property of our cut-off function, \[|\partial_t \phi_{x_i,R}|\leq c_0\frac {\phi_{x_i,R}^\rho} {T}, \]as well as the fact, \[\frac 1 T \leq \frac \nu {R^2}=Pr \frac \eta {R^2},\]we conclude that,
\begin{align*}\frac 1 2 \int_0^T\int(|u|^2+|b|^2)\partial_t\phi_{x_i,R}~dx~dt  &\leq c_0 \frac {\nu} {R^2} \int_0^T \int |u|^2\phi_{x_i,R}^\rho~dx~dt + {c_0} {Pr} \frac {\eta} {R^2}\int_0^T\int |b|^2\phi_{x_i,R}^\rho ~dx~dt.
\end{align*}
Our cut-off functions also satisfy, \[|\Delta \phi_{x_i,R}|\leq c_0 \frac {\phi^{2\rho-1}} {R^2},\]and, consequently,
\begin{align*}\frac  1 2\int_0^T\int(\nu|u|^2+\eta |b|^2)\Delta \phi_{x_i,R} ~dx~dt&\leq c_0 \frac {\nu} {R^2} \int_0^T \int |u|^2\phi_{x_i,R}^{2\rho-1}~dx~dt + {c_0}   \frac {\eta} {R^2}\int_0^T\int |b|^2\phi_{x_i,R}^{2\rho-1} ~dx~dt.
\end{align*}
Noting that $\rho>2\rho-1$ we have, upon combining the above estimates, that,
\begin{align*} F_{x_i,R}^E\geq  \int_0^T\int (\nu |\nabla u |^2+\eta |\nabla b |^2)\phi_{x_i,R} ~dx~dt -\frac {c_0} { R^2}\int_0^T\int (\nu |u|^2+\eta (1+Pr) |b|^2)\phi_{x_i,R}^{2\rho-1}~dx~dt.
\end{align*}
Using Lemma \ref{lemma:ensembleaverages}, we observe that,
 \[ K_2 E_0\geq \bigg\langle \frac 1 T \int_0^T\frac 1 {R^3}\int (\nu |\nabla u |^2+\eta |\nabla b |^2)\phi_{x_i,R} ~dx~dt \bigg\rangle_R \geq \frac {1} {K_1} E_0,\]
and,
\begin{align*}\bigg\langle \frac 1 {T} \int_0^T \frac 1 {R^3}\int \frac 1 2 (\nu|u|^2+\eta(1+Pr)|b|^2)\phi_{x_i,R}^{2\rho-1}~dx~dt\bigg\rangle_R&\leq \nu K_2e_0^u+\eta (1+Pr) K_2e_0^b
\\&\leq \nu K_2 \big(1+Pr^{-1})e_0.
\end{align*}

We can thus interpolate the ensemble average between integral scale quantities as,
\begin{align*} \frac {1} {K_1} E_0 - \nu \frac {c_0 K_2 (1+Pr^{-1})} {R^2}  {e_0}\leq \big\langle F_{}^E \big\rangle_R \leq {K_2} E_0 + \nu \frac {c_0 K_2 (1+Pr^{-1})} {R^2}  {e_0}.
\end{align*}
It is worth remarking that the upper bound follows in virtue of the assumed regularity (i.e. equality in \eqref{ineq:generalizedEnergy}) and is the only place where this assumption is used.  In particular, the lower bound holds for non-regular solutions.

Continuing, we now specify a value for $\beta$, the modification to the inertial range, to be, \[\beta = \bigg(\frac {1}{2c_0 K_1K_2(1+Pr^{-1})}\bigg)^{1/2}.\] 
Because $R$ lies in the inertial range -- i.e. $\tau/\beta\leq R \leq R_0$ -- we have, \[\nu \frac {c_0 K_2 (1+Pr^{-1})} {R^2}  {e_0} \leq \frac 1 {2K_1}E_0.\]
The final lower bound for the ensemble average is thus,
\[ \big\langle F_{}^E \big\rangle_R \geq  \frac {1} {2K_1} E_0.\]
The upper bound follows trivially with our definition of $\beta$ and we conclude that, \begin{align*}\frac {1} {2K_1} E_0 \leq \big\langle   F_{}^E \big\rangle_R \leq  2K_2 E_0.\end{align*}
\end{proof}

\begin{remark}
The condition triggering the cascade is essentially a requirement that the \emph{gradients} of the velocity and the magnetic fields are large (averaged, over the integral domain) with respect to the fields themselves; this will hold in the regions of high \emph{spatial complexity} of the flow (the correction parameter $\beta$ depends on certain \emph{a priori} bounded quantities; however, none of these involve gradients).
\end{remark}

\begin{remark}The dependence of the length of the inertial range on $Pr$ has consequences for when the above result is most physically relevant.  The correction to the Taylor micro-scale is minimized when $Pr^{-1}\lesssim 1$, that is, when $\eta \lesssim \nu$.  This is exactly the scenario for turbulence of plasmas in astronomical settings such as the solar wind and the interstellar medium.
\end{remark}

We briefly discuss how to obtain a result for possibly non-regular suitable weak solutions.  As noted in the proof, the assumption of regularity was only used in establishing the upper bound on the ensemble average. This is the issue we now address.  The physical cause of a strict inequality in the generalized energy inequality is interpreted as the loss of energy due to possible singularities.  This will be denoted by $F^\infty_{\phi}$ (or $F^\infty_{x_i,R}$ if $\phi$ is the cut-off for a $(K_1,K_2)$-cover element) and is defined as the value that fills in the inequality \eqref{ineq:generalizedEnergy},
 \begin{align}\label{eq:generalizedEnergySingular} &\int_0^T\int ( \nu |\nabla u|^2+\eta |\nabla b|^2)\phi~dx~dt
 \\\notag&= \frac  1 2\int_0^T\int(|u|^2+|b|^2)\phi_t~dx~dt  + \frac  1 2\int_0^T\int(\nu|u|^2+\eta |b|^2)\Delta \phi ~dx~dt
\\\notag &\qquad+\frac  1 2\int_0^T\int (|u|^2+2|b|^2+2 p)(u\cdot\nabla \phi)~dx~dt
 -\int_0^T \int (u\cdot b)(b\cdot\nabla \phi)~dx~dt-F^\infty_{\phi}.
\end{align}
%***************************************************************************
To account for the possible strictness of the above inequality we replace the fluxes considered previously with fluxes modified to include the possible loss of energy due to singularities.  For example, in the case of the total energy flux, we establish an interpolative bound on ensemble averages corresponding to the localized {\em modified total energy fluxes}, \[F_\phi^{E,\infty}=F_\phi^E-F_\phi^\infty.\]
Positivity and near-constancy results are then given in terms of the modified flux specified above.  The {\em cascade of total energy modified due to possible non-regularity} is then given by the following theorem, the proof of which is identical modulo a substitution of $F_\phi^{E,\infty}$ for $F_\phi^E$ to the proof of \eqref{thrm:totalCascade}.

\begin{theorem}\label{thrm:totalCascadeNonReg}Assume $u$ and $b$ are suitable weak solutions to 3D MHD.  Let $\{x_i\}_{i=1}^n\subset B(x_0,R_0)$ be the centers of a $(K_1,K_2)$-cover at scale $R$ where the cut-off functions are defined with $T\geq R^2/\nu$.  For a scale- and cover-independent, positive parameter, \[\beta:=\bigg( \frac {1} {2CK_1K_2(1+Pr^{-1})}\bigg)^{1/2},\] where $C$ is a constant determined by structural properties of 3D MHD and our cut-off functions, if
$\tau< \beta R_0,$
then,
 \[\frac {1} {2K_1} E_0 \leq \big\langle  F_{}^{E,\infty} \big\rangle_R \leq  2K_2 E_0,\]
 provided $R$ is contained in the interval
$[\tau /\beta, R_0]$.
\end{theorem}
%*******************************************************************************
%*******************************************************************************
%
% FLUID ENERGY
%
%*******************************************************************************
%*******************************************************************************

\section{Cascade-like dynamics of the total fluid energy}\label{sec:fluid}

Our attention is now turned to establishing conditions under which distinct, inertially transported quantities exhibit cascade-like dynamics.  Note that our terminology is modified to reflect the physics literature wherein the term `cascade' typically refers to an ideally conserved quantity.  We first consider the inertially driven concentration of total fluid energy.  The localized flux quantity of interest is,
\[ \frac 1 2 \int_0^T \int (|u|^2  + 2p)(u\cdot \nabla \phi )~dx~dt =-\int_0^T \int ((u\cdot \nabla)u+p)\cdot (\phi ~u)~dx~dt.\]

\begin{remark}\label{meep}
The total fluid energy flux consists of the kinetic energy flux and the pressure flux,
the local and the non-local parts, respectively. However, the cascade-like results obtained in this section, 
paired with the scale-locality of the flux presented in
Section 7, indicates that the dominant component (in the average) is the local one, i.e.,
the kinetic energy flux. Alternatively, one can study the kinetic energy cascade on its own, and 
try to interpolate the ensemble-averaged pressure flux between suitable integral-scale
quantities.  See Section \ref{sec:indiecascades}.
\end{remark}

Results are ultimately intended for the dimensional form of the equations (as was the case for the cascade of total energy) but it will be convenient to carry out estimates in a \emph{dimensionless} setting.  Before doing so we verify that a positive result in one context implies an analogous result in the other. Beginning with a dimensional problem where the integral domain has radius $R_0$ (the characteristic length scale) and $\nu$ is the kinematic viscosity (these will be our fundamental dimensions; they also determine the characteristic time scale $T= {R_0^2} \nu^{-1}$), non-dimensionalization is achieved using the dimensionless variables and functions,\[x_i^*=\frac {x_i}{R_0},~t^*=\frac t T,~u^*(x,t)=\frac T {R_0} u(x,t),~\mbox{and}~b^*(x,t)=\frac T {R_0} b(x,t).\]
We will establish the existence of a dimensionless cascade-like behavior over an inertial range determined by the relationship,
 \[\bigg( \frac {R_0} R\bigg)^\delta\frac {e_0^*}{E_0^*}<\beta,\]
where,
\begin{align*}e_0^*= \int_0^1\int \frac 1 2 (|u^*|^2+|b^*|^2)\phi_0~dx^*dt^* ~\mbox{and}~E_0^*=\int_0^1\int(|\nabla_* u^*|^2+|\nabla_* b^*|^2)\phi_0 ~dx^*dt^*,
\end{align*}and $\beta$ is a dimensionless constant (here $\nabla_*$ indicates differentiation with respect to the dimensionless variable).
This has consequences for the dimensional setting in virtue of the equality,\[\frac {e_0} {E_0}=R_0^2\frac {e_0^*} {E_0^*},\]which follows from a change of variable and the chain rule.   Similar computations verify the following relationships, \begin{align*}e_0^u=\frac {R_0^2} {T^2} ~e_0^{u^*},~e_0^{b}&= \frac {R_0^2} {T^2} e_0^{b^*},~
E_0^{u}= \frac {\nu} {T^2}E_0^{u^*},~\mbox{and}~
E_0^{b}=\frac {\eta} {T^2} E_0^{b^*},
\end{align*}where the dimensionless quantities $e_0^{u^*}$, $e_0^{b^*}$, $E_0^{u^*}$, and $E_0^{b^*}$ are defined in analogy with their dimensional counterparts.

For the flux presently of interest, total (fluid) energy, setting,
\begin{align*}F_{x_i,R}&=\frac 1 {TR^3}\int_0^T\int ((u\cdot \nabla)u +\nabla p)\cdot (\phi_{x_i,R} u)~dx~dt,\end{align*}and,\begin{align*}
F_{x_i,R}^*&=\frac {R_0^3} {R^3}\int_0^1\int ((u^*(x^*,t^*)\cdot \nabla_*)u^*(x^*,t^*)+\nabla_* p^*(x^*,t^*))\cdot \phi_{x_i,R}(x^*,t^*)u^*(x^*,t^*)~dx^*~dt^*,
\end{align*}we see that, \[F_{x_i,R}=\frac {\nu } {T^2} F_{x_i,R}^*.\]

The equivalence of non-dimensional and dimensional cascade can be seen by considering an example theorem in the dimensionless context which asserts that, for certain dimensionless quantities $\beta_u$ and $\beta_b$, if, \[\frac {R_0} {\beta_u} \bigg(\frac {e_0^{u^*}}{E_0^{u^*}} \bigg)^{1/4}<R<R_0 \quad \mbox{and}\quad \frac {R_0} {\beta_b} \bigg(\frac {e_0^{b^*}}{E_0^{b^*}} \bigg)^{1/4}<R<R_0,\]
then, \[\frac 1 {2K_*} E_0^* \leq \langle F^*\rangle_{R/R_0} \leq 2K_*E_0^*.\]

By the quantitative relations identified above, the consequence for the dimensional scenario is, if,
\[\frac {R_0^{1/2}} {\beta_u} \bigg(\frac {\nu \,e_0^u}{E_0^u} \bigg)^{1/4}<R<R_0 \quad \mbox{and}\quad \frac {R_0^{1/2}} {\beta_b} \bigg(\frac {\eta \,e_0^b}{E_0^b} \bigg)^{1/4}<R<R_0,\]
then, \[\frac 1 {2K_*}  \widetilde{E}_0 \leq \langle F \rangle_R \leq 2  K_* \widetilde{E}_0,\]where $\tilde E_0$ is determined at the scale of the integral domain -- in particular, it is independent of the scale and choice of a $(K_1,K_2)$-cover --  by, \[\widetilde E_0= \bigg(E_0^u+Pr~ E_0^b \bigg).\]

Finally, note that in the dimensionless variables our refined cut-off functions localize spatially to balls of (dimensionless) radius $R/R_0$ for $0<R\leq R_0$ and temporally to the (dimensionless) interval $[0,1]$ and, in analogy to properties \eqref{timecutoff} and \eqref{spacecutoff}, satisfy the following gradient estimates,
\begin{align}\label{ineq:dimensionlessCutOffs}|\nabla_* \phi (x^*,t^*)|\leq  c_0\frac {R_0} {R} \phi^\rho(x^*,t^*) \quad \mbox{and}\quad |\phi_{t^*}(x^*,t^*)|\leq c_0\phi^{\rho}(x^*,t^*)
.\end{align}

%*****************************************************************************
%*****************************************************************************
%*****************************************************************************

Since we can recover dimensional cascade-like behavior from the non-dimensional counterpart we are justified in considering only the latter.  We subsequently suppress the asterisks used above to indicate non-dimensionality noting that, for the remainder of this section, we are working with dimensionless quantities.

Following \cite{SeTe}, dimensionless 3D MHD is,
\begin{align*}\partial_t u-\frac 1 {Re} \Delta u &= -(u\cdot \nabla)u -\nabla p - S\bigg(\nabla \frac {|b|^2} 2 - (b\cdot \nabla)b\bigg),
\\ \partial_t b-\frac 1 {Rm} \Delta b &= \nabla \times (b\times u),
\\ \nabla\cdot u&=\nabla \cdot b =0,
\\ u(x,0)&=u_0(x) \in L^2(\R^3),
\\b(x,0)&=b_0(x)\in L^2(\R^3),
\end{align*} where $Re$ and $Rm$ are the Reynolds and magnetic Reynolds numbers respectively and the non-dimensional number, $S$, is defined in terms of the Hartmann number, $M$, to be $S=M^2/(ReRm)$ (note that the Hartmann number is a dimensionless quantity given by the ratio of the electromagnetic force to the viscous force).

In order to obtain a local cancellation between coupled non-linear terms, we assume sufficient regularity so that the following derivation is justified. Taking the scalar product of the equation of motion by $\phi u$ and the induction equation by $S \phi b$, integrating, and rearranging we obtain a flux density for total (fluid) energy into the ball $B(x_i,R/R_0)$ (here $\phi$ denotes an appropriate refined cut-off function),
\begin{align}\notag \int_0^1\int \bigg( \frac 1 2 |u|^2 + p\bigg)\cdot (\nabla \phi\cdot u)~dx~dt&=\frac 1 {Re} \int_0^1 \int |\nabla u|^2\phi~dx~dt + \frac S {Rm} \int_0^1\int |\nabla b|^2\phi~dx~dt
\\\notag&\quad-S\int_0^1\int \nabla \phi \cdot (b\times (u\times b))~dx~dt
\\\notag&\quad- \frac 1 2 \int_0^1\int |u|^2\phi_t~dx~dt- \frac 1 {2Re} \int_0^1 \int |u|^2\Delta \phi~dx~dt
\\\notag&\quad-\frac S 2 \int_0^1\int |b|^2\phi_t~dx~dt- \frac S {2Rm} \int_0^1 \int |b|^2\Delta \phi~dx~dt.
\end{align}
Bounds for the lower order terms on the right hand side follow.  The last four terms are bounded in a manner similar to that seen in the proof of \eqref{thrm:totalCascade}.  Here, however, we acknowledge the change of variable and cite the properties of our refined cut-off functions,  \eqref{ineq:dimensionlessCutOffs}, as well as the chain rule, to obtain,
\begin{align}\label{ineq:uEnergyWithReynolds}\bigg|\frac 1 2 \int_0^1\int |u|^2\phi_t~dx~dt\bigg| &\leq c_0Re \frac 1 {Re} \bigg( \frac {R_0} R\bigg)^2 \int_0^1\int |u|^2\phi^{4\rho-3} ~dx~dt,
\\\notag \bigg|\frac 1 {2Re} \int_0^1 \int |u|^2\Delta \phi~dx~dt\bigg|&\leq c_0  \frac 1 {Re} \bigg( \frac {R_0} R\bigg)^2 \int_0^1\int |u|^2\phi^{4\rho-3} ~dx~dt,
%**************************************************************
\\\notag \bigg|\frac S 2 \int_0^1\int |b|^2\phi_t~dx~dt \bigg|&\leq c_0 Rm \frac S {Rm} \bigg( \frac {R_0} R\bigg)^2\int_0^1\int |b|^2\phi^\rho ~dx~dt,
\\\notag \bigg|\frac S {2Rm} \int_0^1 \int |b|^2\Delta \phi~dx~dt \bigg|&\leq c_0\frac S {Rm} \bigg( \frac {R_0} R\bigg)^2  \int_0^1\int |b|^2\phi^{4\rho-3} ~dx~dt.
\end{align}
Repeatedly using H\"older's inequality, the Gagliardo-Nirenberg inequality, and Young's inequality, the following gives a bound for the non-linear term originating in the induction equation,
\begin{align}\notag \bigg| c_0 S\frac {R_0} R \int_0^1\int \phi^{\rho}(b\times (u\times b))~dx~dt\bigg|
&\leq c_0 S\frac {R_0} {R^{}} \int_0^1\int (|u|^{1/2}\phi^{\rho-3/4})(|b|^{1/2})(|u|^{1/2}|b|^{3/2}\phi^{3/4})~dx~dt
\\\notag &\leq c_0 S \frac {R_0} {R^{}} \int_0^1 ||u\phi^{2\rho-3/2}||_2^{1/2}||b||_2^{1/2} || (|u|^{1/2}|b|^{3/2}\phi^{3/4})||_2~dt
\\\notag &\leq c_0 \frac {R_0} {R^{}} \bigg(\frac {M^{1/2}Rm^{1/2}} {Re^{1/4}} \bigg) \bigg( \frac {M^{3/2}} {Re^{3/4}Rm^{3/2}}\bigg)
\\\notag &\qquad \cdot \int_0^1 ||u\phi^{2\rho-3/2}||_2^{1/2}||b||_2^{1/2} ||u||_2^{1/2}||\nabla (b\phi^{1/2})||_2^{3/2}~dt
\\\notag &\leq c_0\frac {M^2Rm^2R_0^4} {Re R^{4}} \big( \sup_t ||u||_2^2\big)  \big( \sup_t ||b||_2^2\big)\int_0^1 ||u\phi^{4\rho-3}||_2^2~dt
\\\notag &\qquad + \frac S {4Rm} \int_0^1 ||\nabla (b\phi^{1/2})||_2^2~dt
\\\notag &\leq c_0 (MRm)^2\big( \sup_t ||u||_2^2  \sup_t ||b||_2^2\big)\frac 1 {Re}\frac {R_0^4 } {R^{4}}\int_0^1 ||u\phi^{4\rho-3}||_2^2~dt \\\label{ineq:tripleCrossProductTerm}&\qquad + \frac S {4Rm} \int_0^1 ||(\nabla b)\phi^{1/2}||_2^2~dt + \frac {c_0S} {Rm }\frac {R_0^4} {R^4}\int_0^1 ||b\phi^{4\rho-3}||_2^2~dt.
%**************************************************************
\end{align}

Taking ensemble averages and passing  where appropriate (via Lemma \ref{lemma:ensembleaverages}) to the quantities $e_0^u$, $e_0^b$, $E_0^u$, and $E_0^b$, we have ,
\begin{align*}\langle F \rangle_R &\geq \frac 1 {K_1}  \frac 1 {Re} E_0^u - c_0 \frac {K_2} {Re}\big(1+Re + (MRm)^2\big( \sup_t ||u||_2^2  \sup_t ||b||_2^2\big)\big) \frac {R_0^4 } {R^{4}}e_0^u
\\&\qquad+ \frac 1 {K_1} \frac {S} {Rm} E_0^b - c_0 \frac {K_2S}{Rm} \big(2+Rm \big)  \frac {R_0^4 } {R^{4}}e_0^b
\\&= \frac 1 {K_1}  \frac 1 {Re} E_0^u -C_u\frac {K_2} {Re} \frac {R_0^4 } {R^{4}} e_0^u + \frac 1 {K_1} \frac {S} {Rm} E_0^b -C_b\frac {K_2 S} {Rm} \frac {R_0^4 } {R^{4}} e_0^b,\end{align*}
where in the last line we have set, \[C_u=c_0\bigg( 1+Re + (MRm)^2\big( \sup_t ||u||_2^2  \sup_t ||b||_2^2\big)\bigg)~\mbox{and}~ C_b=c_0\big(2+Rm \big).\]
Similarly, an upper bound is,
\begin{align*}\langle F \rangle_R &\leq K_2  \frac 1 {Re} E_0^u + C_u K_2\frac {1} {Re} \frac {R_0^4 } {R^{4}} e_0^u+  K_2 \frac {S} {Rm} E_0^b +C_b K_2 \frac { S} {Rm}\frac {R_0^4 } {R^{4}} e_0^b.\end{align*}

Defining now the parameters for a correction to the inertial range by, \[\beta_u=\bigg( \frac 1 {2K_1K_2C_u}\bigg)^{1/4}~\mbox{and}~\beta_b=\bigg( \frac 1 {2K_1K_2C_b}\bigg)^{1/4},\]we have justified the following theorem.
\begin{theorem}\label{thrm:totalFluidCascadeNonD}Let $\{x_i\}_{i=1}^n\subset B(x_0,R_0)$ be the centers of a $(K_1,K_2)$-cover at scale $R$.  For $\beta_u$ and $\beta_b$ defined above, if, \[\tau:=\max\bigg\{  \bigg( \frac {e_0^u} {E_0^u}\bigg)^{1/4},  \bigg(\frac {e_0^b} {E_0^b}\bigg)^{1/4} \bigg\}<\min\{ \beta_u,\beta_b\}=:\beta, \]
then for scales $R$ where $\tau/\beta \leq R/ R_0,$ we have,
 \[\frac {1} {2K_1 }\bigg( \frac 1 {Re} E_0^u +\frac S{Rm} E_0^b \bigg)\leq \langle F_{}^{} \rangle_{R}\leq  2K_2\bigg( \frac 1 {Re} E_0^u +\frac S{Rm} E_0^b\bigg).\]
\end{theorem}

\begin{remark}The above is particularly relevant in scenarios where $S/Rm\simeq 1$ and the integral scale magnetic energy dominates the integral scale kinetic energy.  In this case it is possible to free the parameter $\beta_u$ from its dependence on $Re$ by substituting for the bound \eqref{ineq:uEnergyWithReynolds}, the bound \[\frac 1 2 \int_0^1 \int |u|^2\partial_t\phi_{x_i,R}~ dx~dt\leq K_2e_0^b.\]   This allows for very large values of $Re$ and is applicable to settings involving the confinement of a liquid metal by a strong magnetic guide field \cite{VoDaKn05}. Without these assumptions on the flow the physical relevance of the above result is diminished by the fact that the dependencies of $\beta_u$ and $\beta_b$ on the fluid and magnetic Reynold's numbers result in an inertial range which decreases in length as $Rm$ and $Re$ increase.
\end{remark}

%*****************************************************************************
%
% Individual Cascades.
%
%*****************************************************************************

\section{Other cascade-like dynamics}\label{sec:indiecascades}
We presently identify conditions under which cascade-like dynamics are exhibited by individual energy exchange mechanisms -- including the exchange of energy between the velocity and magnetic fields -- are evident in the context of the dimensionless equations.  Throughout this section all solutions are assumed to be suitable weak solutions to the dimensionless 3D MHD system, regular enough for the localized energy equality to hold.

The localized (by the scalar function $\phi$ to a ball of radius $R/R_0$) flux term responsible for the kinetic to kinetic energy exchange driven by the advection of the velocity field is,
\begin{align}\label{term:uTouEnergy}F_{x_i,R}^u:= -\int_0^1\int(u\cdot \nabla )u \cdot( \phi~u)~dx~dt =\frac 1 2 \int_0^1\int |u|^2(u\cdot \nabla\phi)~dx~dt.\end{align}

Similarly, the (localized) magnetic to magnetic energy transfer driven by the advection of the velocity field is,
\begin{align}\label{term:bTobEnergy}F_{x_i,R}^b:=-\int_0^1\int (u\cdot \nabla b) \cdot (\phi ~b)~dx~dt =\frac 1 2 \int_0^1\int |b|^2(u\cdot \nabla \phi)~dx~dt.\end{align}

The fluid pressure flux-type term is, \begin{align}\label{term:pressureFlux}F_{x_i,R}^p:=-\int_0^1 \int \nabla p \cdot (\phi~ u)~dx~dt&=\int_0^1 \int p ~(u \cdot \nabla \phi)~dx~dt.
\end{align}

As already mentioned, the transfer of magnetic to kinetic energy is driven by the Lorentz force while the stretching of the magnetic field lines is responsible for the transfer of kinetic energy to magnetic energy; combined (locally), they yield the following term -- the advection of cross-helicity by the magnetic field,

\begin{align}\label{term:crossEnergy}F_{x_i,R}^{ub}:= \int_0^1 \int  ((b\cdot \nabla )b \cdot (\phi~u ) + (b\cdot\nabla)u \cdot (\phi ~b ))~dx~dt= -\int_0^1\int (u\cdot b )(b\cdot \nabla \phi)~dx~dt.
\end{align}

Until now we have only investigated cascades associated with collections of flux-type terms including the term for the flux of the fluid pressure.  Due to the unique structure of this term, additional effort is required to establish cascade-like behavior for combined fluxes excluding the pressure flux. More precisely, we will need to bound the quantity, \[\int_0^1 \int p~ (u\cdot \nabla \phi)~dx~dt.\]
To do so we will use the estimate,
\begin{align}\notag ||u\phi^{1/2}||_3&=\bigg( \int (|u|^{3/2}  \phi^{3/4})(|u|^{3/2}\phi^{3/4})~dx\bigg)^{1/3}
\\\notag &\leq \bigg( \bigg(\int |u|^2\phi~dx \bigg)^{3/4} \bigg( \int |u|^6\phi^3~dx \bigg)^{1/4} \bigg)^{1/3}
\\\notag&\leq C ||u\phi^{1/2}||_2^{1/2}||\nabla(u\phi^{1/2})||_2^{1/2},
\end{align}
which allows that,
\begin{align}&\notag \bigg|\int_0^1 \int p~ (u\cdot \nabla \phi)~dx~dt\bigg|
\\&\leq C \frac {R_0} R \int_0^1 ||p\phi^{\rho-1/2}||_{3/2} || u \phi^{1/2}||_3~dt
\\\notag &\leq  C \frac {R_0} R \int_0^1 ||p\phi^{\rho-1/2}||_{3/2} || u \phi^{1/2}||_2^{1/2} ||\nabla (u\phi^{1/2})||_2^{1/2}~dt
\\\notag &\leq C \frac {R_0} R \bigg(\int_0^1 ||p\phi^{\rho-1/2}||_{3/2}^{3/2} ~dt \bigg)^{2/3} \bigg(\int_0^1  || u \phi^{1/2}||_2^{6} ~dt \bigg)^{1/12} \bigg(\int_0^1  ||\nabla (u\phi^{1/2})||_2^{2} ~dt \bigg)^{1/4}
\\\notag &\leq   C Re^{1/3} \bigg(\frac {R_0} R \bigg)^{4/3} \bigg(\int_0^T ||p\phi^{\rho-1/2}||_{3/2}^{3/2}~dt\bigg)^{8/9} \bigg( \int_0^1 || u \phi^{1/2}||_2^{6}~dt\bigg)^{1/9}
\\\notag &\qquad + \frac 1 {8 Re} \int_0^1||\nabla (u\phi^{1/2})||_2^2~dt
\\ \notag&\leq   {C_p} Re^{1/3} \bigg(\frac {R_0} R \bigg)^{4/3} \bigg( \int_0^1 || u \phi^{1/2}||_2^{2}~dt\bigg)^{1/9}+ C \frac 1 {Re} \bigg(\frac {R_0} R \bigg)^{2} \int_0^1 ||u \phi^{2\rho-1}||_2^2~dt
\\ &\qquad
 +  \frac 1 {4Re}  \int_0^1 ||(\nabla u) \phi^{1/2}||_2^2~dt,\label{ineq:pressurebound}
\end{align}
where we have used H\"older's inequality, the Gagliardo-Nirenberg inequality, and Young's inequality.  Note that the quantity appearing above, 

\begin{equation}\label{cp}
 C_p=\big( \sup_t||u\phi_0^{1/2}||_2\big )^{4/9}\bigg(\int_0^1 ||p\phi_0^{\rho-1/2}||_{3/2}^{3/2}~dt\bigg)^{8/9}, 
 \end{equation}
is dimensionless and \emph{a priori} bounded in virtue of regularity properties of suitable weak solutions (cf. \cite{HeXin1}). In addition, it contains no gradients and is independent of the particular cover element at scale $R$.

Regarding the first term on the right hand side of \eqref{ineq:pressurebound}, in order to pass from ensemble averages of localized quantities to an integral scale quantity we will need a simple consequence of the finite form of Jensen's inequality.  Specifically, for $\{a_i\}_{i=1}^n$ a set of non-negative values,
\begin{align*} \sum_{i=1}^n \frac {a_i^{1/9}} n\leq \bigg( \sum_{i=1}^n \frac {a_i} n \bigg)^{1/9}. \end{align*}
Taking an ensemble average of normalized quantities yields,
\begin{align}\notag \frac {C_p} n \sum_{i=1}^n \bigg(\frac {R_0} {R}\bigg)^3  \bigg(\frac {R_0} R \bigg)^{\frac 4 3} \bigg( \int_0^1 \int |u|^2 \phi_{x_i,R}~dt\bigg)^{\frac 1 9} &\leq  C_p \bigg(\frac {R_0} R \bigg)^4 \bigg( \frac 1 n \sum_{i=1}^n  \int_0^1  \bigg(\frac {R_0} R \bigg)^3 \int |u|^2 \phi_{x_i,R}~dx~dt \bigg)^{\frac 1 9}
\\   &\leq  {C_p}  \bigg(\frac {R_0} R \bigg)^4 \big(K_2 e_0^u\big)^{\frac 1 9} \label{ineq:pressureEnsAv}.\end{align}

Bounds for the remaining flux densities are,
\begin{align}\notag \big|F_{x_i,R}^u\big| &\leq C\frac {R_0} {R} \int_0^1 ||u||_2 ||u\phi^{2\rho-3/2}||_2^{1/2}||\nabla (u\phi^{1/2})||_2^{3/2}~dx~dt
\\\notag &\leq C\bigg( \frac {R_0 } {  R }\bigg)^{4} \sup_t ||u||_2^4 \int_0^1  ||u\phi^{2\rho-3/2}||_2^2~dt + \frac 1 {8Re} \int_0^1   ||\nabla (u\phi^{1/2})||_2^2~dt
\\&\notag \leq C \big(Re^{4/3}\sup_t ||u||_2^4 +1\big)\frac 1 {Re}\bigg( \frac {R_0 } {  R }\bigg)^{4}  \int_0^1  ||u\phi^{2\rho-3/2}||_2^2~dt
\\&\notag\qquad + \frac 1 {4Re} \int_0^1  ||(\nabla  u)\phi^{1/2}||_2^2~dt,
\end{align}and,
\begin{align}\big|F_{x_i,R}^b\big|,~\big|F_{x_i,R}^{ub}\big|&\leq  C \big(Rm^{4/3}\sup_t ||u||_2^4 +1\big)\frac S {Rm}\bigg( \frac {R_0 } {  R }\bigg)^{4}  \int_0^1  ||b\phi^{2\rho-3/2}||_2^2~dt
\\&\notag\qquad + \frac S {4Rm} \int_0^1  ||(\nabla  b)\phi^{1/2}||_2^2~dt.
\end{align}

Henceforth, we will be concerned with the scenario in which the kinetic energy over the integral domain and within the prescribed time-scale remains bounded away from zero. In particular we take this to mean $e_0^u\geq 1$. This is the trade-off for considering \emph{purely kinetic} fluxes.
The result concerning cascade-like behavior for the direct kinetic to kinetic energy transfer driven by the advection of the velocity field follows.

\begin{theorem}\label{thrm:FluidEnergyCascadeNonD}Assume $u$ and $b$ are solutions of 3D MHD possessing sufficient regularity so that the generalized energy equality holds.  Let $\{x_i\}_{i=1}^n\subset B(x_0,R_0)$ be the centers of a $(K_1,K_2)$-cover at scale $R$.  Suppose that the flow is such that $e_0^u\geq 1$. For certain values $\beta_u$ and $\beta_b$, if, \[\tau:=\max\bigg\{  \bigg( \frac {e_0^u} {E_0^u}\bigg)^{1/4},  \bigg(\frac {e_0^b} {E_0^b}\bigg)^{1/4} \bigg\}<\min\{ \beta_u,\beta_b\}=:\beta, \]
then, for scales $R$ where $\tau/\beta \leq R/ R_0,$ we have,
 \[\frac {1} {2K_1 }\bigg( \frac 1 {Re} E_0^u +\frac S{Rm} E_0^b \bigg)\leq \langle F_{}^{u} \rangle_{R}\leq  2K_2\bigg( \frac 1 {Re} E_0^u +\frac S{Rm} E_0^b\bigg).\]
\end{theorem}
\begin{proof}The main estimates have already been established.  Based on the localized energy equality we have
\begin{align*}  F_{x_i,R}^u  &\geq  \frac {1} {Re} E_{x_i,R}^u +\frac { S} {Rm} E_{x_i,R}^b - \big| F^p_{x_i,R}+N_{x_i,R}\big|
\\&\quad- \frac 1 2 \int_0^1\int |u|^2\phi_t~dx~dt- \frac 1 {2Re} \int_0^1 \int |u|^2\Delta \phi~dx~dt
\\&\quad-\frac S 2 \int_0^1\int |b|^2\phi_t~dx~dt- \frac S {2Rm} \int_0^1 \int |b|^2\Delta \phi~dx~dt.
\end{align*}
By the assumption on the flow we can replace the exponent of $1/9$ in the bound \eqref{ineq:pressureEnsAv} with $1$. Taking ensemble averages, we obtain the lower bound,
\begin{align*}\langle F^u\rangle_R&\geq  \frac {1} {2K_1Re} E_{x_i,R}^u  - C_u  K_2 \frac 1 {Re} \bigg(\frac {R_0} {R}\bigg)^4 e_0^u +\frac { S} {K_1 Rm} E_{x_i,R}^b - C_b K_2 \frac {S} {Rm}\bigg(\frac {R_0} {R}\bigg)^4 e_0^b,
\end{align*}
where \[C_u=C\big(C_pRe^{4/3}+Re+ 1+ (MRm)^2\sup_t ||u||_2^2\sup_t ||b||_2^2\big),\]
and, \[C_b=C \big( Rm^{4/3}\sup_t ||u||_2^4 +Rm+1\big).\]This is sufficient to establish values for $\beta_u$ and $\beta_b$ (containing no gradients) and conclude in the standard fashion.
\end{proof}

At this point, noting that $F^{b}$ and $F^{ub}$ both satisfy \eqref{ineq:tripleCrossProductTerm}, we have already demonstrated the steps involved in establishing cascade-like behavior for the densities $F^b$ and $F^{ub}$.  We consequently omit the proofs. 
\begin{theorem} Let $\{x_i\}_{i=1}^n\subset B(x_0,R_0)$ be the centers of a $(K_1,K_2)$-cover at scale $R$.  Suppose that the flow is such that $e_0^u\geq 1$. For certain values $\beta_u$ and $\beta_b$, if, \[\tau:=\max\bigg\{  \bigg( \frac {e_0^u} {E_0^u}\bigg)^{1/4},  \bigg(\frac {e_0^b} {E_0^b}\bigg)^{1/4} \bigg\}<\min\{ \beta_u,\beta_b\}=:\beta, \]
then for scales $R$ where $\tau/\beta \leq R/ R_0,$ we have,
 \[\frac {1} {2K_1 }\bigg( \frac 1 {Re} E_0^u +\frac S{Rm} E_0^b \bigg)\leq \langle F_{}^{b},F^{ub} \rangle_{R}\leq  2K_2\bigg( \frac 1 {Re} E_0^u +\frac S{Rm} E_0^b\bigg).\]
\end{theorem}

\section{Locality of the flux}
\label{sec:locality}
According
to turbulence phenomenology in the purely hydrodynamical setting,
the average energy flux at scale $R$ is supposed to be \emph{well-correlated}
only with the average fluxes at \emph{nearby scales} (throughout the 
inertial range). This phenomenon has been confirmed in
\cite{E05, CCFS08, DaGr1, DaGr2}.
In the plasma setting, the question of locality has been somewhat
controversial. Recently, Aluie and Eyink \cite{AlEy10} produced an argument
in favor of locality of the total energy and cross-helicity fluxes. Numerical work also supports locality (cf. \cite{DVC-05} for the case of decaying turbulence).

Our context allows us to affirm a particular flux's locality as a direct consequence
of the existence of the corresponding, nearly-constant turbulent cascade
\emph{per unit mass}. 
We illustrate this in the case of the kinetic energy flux
(transported by the velocity).
Denote the time-averaged local fluxes associated to the cover
element $B(x_i,R)$ by $\hat{\Psi}_{x_i,R}$,
\[
\hat{\Psi}_{x_i,R} = \frac{1}{T} \int_0^T  \int
\frac{1}{2}|u|^2 (u\cdot \nabla\phi_i) \, dx,
\]
and the time-averaged local fluxes associated to the cover
element $B(x_i,R)$, \emph{per unit mass}, by $\hat{\Phi}_{x_i,R}$,
\[
\hat{\Phi}_{x_i,R} = \frac{1}{T} \int_0^T  \frac{1}{R^3} \int
\frac{1}{2}|u|^2 (u\cdot \nabla\phi_i) \, dx.
\]
Then, the (time and ensemble) averaged flux is given by,
\[
 \langle\Psi\rangle_R = \frac{1}{n}\sum_{i=1}^n \hat{\Psi}_{x_i,R} =
 R^3 \, \langle \Phi \rangle_R = R^3 \frac{1}{n}\sum_{i=1}^n \hat{\Phi}_{x_i,R}.
\]
 
The following manifestation of locality follows directly from
Theorem \ref{thrm:FluidEnergyCascadeNonD}.
Let $R$ and $r$ be
two scales within the inertial range delineated in the theorem. Then,
\[
\frac{1}{4K_1^2} \biggl(\frac{r}{R}\biggr)^3 \le \frac{\langle \Psi
\rangle_r}{\langle \Psi \rangle_R} \le 4K_2^2
\biggl(\frac{r}{R}\biggr)^3.
\]

In particular, if $r=2^k R$ for some integer $k$, 
\[
\frac{1}{4K_1^2} \ 2^{3k} \le \frac{\langle \Psi \rangle_{2^k
R}}{\langle \Psi \rangle_R} \le 4K_2^2 \ 2^{3k},
\]
i.e., along the \emph{dyadic scale}, the locality propagates exponentially.

%****************************************************************************
\section{A scenario exhibiting predominant $u$-to-$b$ energy transfer}
\label{sec:scenario}
%****************************************************************************
In \cite{DaGr5} a dynamic estimate is given on the vortex-stretching term (in the vorticity formulation of 3D NSE) -- across a range of scales -- using the ensemble averaging process we have illustrated above.  The purpose was to present a mathematical evidence of the creation and persistence of integral scale length vortex filaments by establishing positivity of the ensemble-averaged vortex stretching term across a range of scales extending to the integral scale.

In the induction equation for the magnetic field, the nonlinear term $(b\cdot \nabla)u$ is responsible for the stretching of magnetic field lines.  Positivity of $(b\cdot \nabla)u\cdot(\phi ~b)$ indicates the magnetic field line is being elongated, a phenomenon which corresponds to a transfer of energy from the velocity field to the magnetic field (negativity would reflect a diminution of the field line and a local transport of energy from the magnetic field to the fluid flow). Consequently, to conclude that the predominant energy exchange between the velocity and magnetic fields is from the velocity field to the magnetic field across a range of physical scales, it will be sufficient to establish (in an appropriate statistical sense) the positivity of $(b\cdot \nabla)u\cdot(\phi ~b)$ across these scales.  Before proceeding to this task we remark that recent numerical work (cf. \cite{DVC-05}) indicates that imbalanced exchanges are common in certain forced and decaying turbulent regimes.

We label the space-time localized quantity of interest as,
\[V_{x_i,R}=\int_0^1\int (b\cdot\nabla)u\cdot (b~\phi_{x_i,R})~dx~dt,\]and note that our work is carried out in the context of the dimensionless formulation of 3D MHD where we take $Rm=S=1$ for convenience. We also assume the weak solution in question is regular. 

\begin{theorem} Let $\{x_i\}_{i=1}^n\subset B(x_0,R_0)$ be the centers of a $(K_1,K_2)$-cover at scale $R$.   For a certain value (which will be apparent in the proof) $\beta>0$, if, \[\tau:=   \bigg(  \frac {e_0^b} {E_0^b}\bigg)^{1/4}  < \beta, \]
then for scales $R$ where $\tau/\beta \leq R/R_0,$ we have,
 \[\frac {1} {2K_1 }E_0^b\leq \langle V_{x_i,R}  \rangle_{R}\leq  2K_2  E_0^b .\]
\end{theorem}

\begin{proof}Starting with the induction equation it is routine to obtain,
\begin{align*}\int_0^1\int (b\cdot \nabla)u\cdot (b \phi_{x_i,R} )~dx~dt&= \int_0^1\int |\nabla b|^2\phi_{x_i,R}~dx~dt -\int_0^1\int \frac 1 2 |b|^2(\partial_t \phi_{x_i,R}-\Delta \phi_{x_i,R})~dx~dt
\\&\qquad + \frac 1 2 \int_0^1\int |b|^2(u\cdot \nabla \phi_{x_i,R})~dx~dt.\end{align*} Note that, for a refined cut-off function $\phi$, the bounding process evident in the derivation of \eqref{ineq:tripleCrossProductTerm} can be modified to yield,
\begin{align*}\int_0^1\int \frac {|b|^2} 2  u\cdot \nabla \phi~dx~dt&\leq C \bigg(\frac {R_0} {R}\bigg)^4 \big(\sup_t ||u||_2 \big)^4\int_0^1 ||b\phi^{}||_2^2~dt+ \frac 1 4 \int_0^1||\nabla (\phi~ b)||_2^2~dt.
\end{align*}Consequently, after taking ensemble averages and recalling the last two estimates in \eqref{ineq:uEnergyWithReynolds},
\begin{align*} \langle V_{x_i,R} \rangle_R&\leq  K_2 E_0^b +K_2\big(\sup_t ||u||_2^4 +2\big)\bigg(\frac {R_0} R \bigg)^2e_0^b
\end{align*}
and,
\begin{align*}\langle V_{x_i,R} \rangle_R&\geq  K_2 E_0^b - \frac C {2K_1} \big(\sup_t ||u||_2^4 +2\big)\bigg(\frac {R_0} R \bigg)^2e_0^b.
\end{align*}Selecting an appropriate value for $\beta$ allows us to conclude in the standard fashion.
\end{proof}

%****************************************************************************************
%****************************************************************************************
%****************************************************************************************
\subsection*{Acknowledgements.}

Z.B. acknowledges the support of the \emph{Virginia Space Grant Consortium} via the Graduate Research Fellowship; Z.G. acknowledges the support of the \emph{Research Council of Norway} via the grant 213474/F20, and the \emph{National Science Foundation} via the grant DMS 1212023.

\bibliographystyle{plain}
\bibliography{references}

\end{document}